\documentclass [11pt] {article} 

\usepackage{fullpage}
\usepackage[pdftex]{graphicx}

\usepackage{color}
\usepackage{amsmath}
\usepackage{amssymb}
\usepackage{amsfonts}
\usepackage{enumitem}

\usepackage[english]{babel}

\usepackage{multirow}
\usepackage{rotating}

\usepackage{epstopdf}

\begin{document}

\newtheorem{theorem}{Theorem}[section]
\newtheorem{lemma}{Lemma}[section]
\newtheorem{corollary}{Corollary}[section]
\newtheorem{claim}{Claim}[section]
\newtheorem{proposition}{Proposition}[section]
\newtheorem{definition}{Definition}[section]
\newtheorem{fact}{Fact}[section]
\newtheorem{example}{Example}[section]

\newcommand{\cA}{{\cal A}}
\newcommand{\cP}{{\cal P}}
\newcommand{\cC}{{\cal C}}
\newcommand{\cG}{{\cal G}}
\newcommand{\cN}{{\cal N}}
\newcommand{\cU}{{\cal U}}
\newcommand{\cT}{{\cal T}}
\newcommand{\cS}{{\cal S}}
\newcommand{\cL}{{\cal L}}
\newcommand{\cV}{{\cal V}}
\newcommand{\loc}{{\cal LOCAL}}
\newcommand{\cM}{{\cal M}}

\newcommand{\f}[4][0pt]{%
\begin{list}{#2}{%
 \setlength{\leftmargin}{#3 em}\addtolength{\leftmargin}{#3 em}\addtolength{\leftmargin}{1 em}%
 \setlength{\labelsep}{#3 em}\addtolength{\labelsep}{#3 em}
 \setlength{\labelwidth}{20pt}
 \if!#1!\else\addtolength{\leftmargin}{#1}\fi%
 \setlength{\topsep}{2pt}
 \setlength{\partopsep}{0pt}%
 \if!#1!\else\setlength{\itemindent}{-#1}\fi%
}%
   \item  #4%
\end{list}%
}

\newcommand{\qed}{\hfill $\square$ \smallbreak}
\newenvironment{proof}[1][Proof]
{\par\noindent{\bf #1:} }{\hspace*{\fill}\nolinebreak{$\Box$}\bigskip\par}

\newcommand{\view}{\mathcal{V}}
\newcommand{\agent}{\lambda}
\newcommand{\lab}{\alpha}
\newcommand{\labels}{\mathcal{L}}
\newcommand{\lista}{\mathcal{Q}}
\newcommand{\code}{\xi}
\newcommand{\home}{h}
\newcommand{\routebegin}{b}
\newcommand{\routeend}{d}
\newcommand{\cR}{\mathcal{R}}
\newcommand{\hist}{\mathcal{H}}
\newcommand{\cQ}{\mathcal{Q}}

\newcommand{\ints}{\mathbb{N}}
\newcommand{\algorithmCL}{{\tt Choose\textup{-}Leader}}
\newcommand{\algorithmUL}{{\tt Update\textup{-}Label}}
\newcommand{\algorithmLE}{{\tt Leader\textup{-}Election}}
\newcommand{\algorithmInit}{{\tt Initialization}}

\newcommand{\last}{\textup{last}}
\newcommand{\depth}{3(n-1)}
\newcommand{\len}{\ell}
\newcommand{\conditionEC}{\textup{EC}}

\title{{\bf  How to Meet Asynchronously at Polynomial Cost}\thanks{A preliminary version of this paper appeared in
Proc. 32nd Annual ACM Symposium on Principles of Distributed Computing (PODC 2013).} }

\author{
Yoann Dieudonn\'{e}\thanks{
MIS, Universit\'{e} de Picardie  Jules  Verne Amiens,  France. E-mail:  {\tt yoann.dieudonne@u-picardie.fr}}
 \and Andrzej Pelc\thanks{
 D\'epartement d'informatique, Universit\'e du Qu\'ebec en Outaouais, Gatineau,
Qu\'ebec J8X 3X7, Canada. {\tt pelc@uqo.ca}. Partially supported by NSERC discovery grant 
and by the Research Chair in Distributed Computing at the
Universit\'e du Qu\'ebec en Outaouais.}
\and  Vincent Villain\thanks{
MIS, Universit\'{e} de Picardie  Jules  Verne Amiens,  France. E-mail: {\tt vincent.villain@u-picardie.fr}}
}

\maketitle

\thispagestyle{empty}

\begin{abstract}

Two mobile agents starting at different nodes of an unknown network have to meet.
This task is known in the literature as {\em rendezvous}.
Each agent has a different label which is a positive integer known to it, but unknown to the other agent.
Agents move in an asynchronous way: the speed
of agents may vary and is controlled by an adversary.  
The cost of a rendezvous algorithm is the total number of edge traversals by both agents until their meeting.
The only previous deterministic algorithm solving this problem has cost exponential in the size of the graph and in the larger label.
In this paper we present a deterministic rendezvous algorithm with cost {\em polynomial} in the size of the graph and in the {\em length}
of the {\em smaller} label. Hence we decrease the cost exponentially in the size of the graph and doubly exponentially in the labels of agents.

As an application of our rendezvous algorithm we solve several fundamental problems involving teams of unknown size larger than 1 of labeled agents
moving asynchronously in unknown networks. Among them are the following problems: {\tt team size}, in which every agent has to find the total
number of agents, {\tt leader election}, in which all agents have to output the label of a single agent, {\tt perfect renaming} in which all agents
have to adopt new different labels from the set $\{1,\dots,k\}$, where $k$ is the number of agents, and {\tt gossiping}, in which each agent has initially a piece of information (value) and all agents have to output all the values. Using our rendezvous algorithm we solve all these problems at cost polynomial in the size of the graph and in the smallest length of all labels of participating agents. 

\vspace*{1cm}

{\bf keywords:} asynchronous mobile agents, network, rendezvous, deterministic algorithm, leader election, renaming, gossiping

\vspace*{3cm}
\end{abstract}

\pagebreak
\section{Introduction} \label{sec:intro}

\noindent
{\bf The background.}
Two mobile agents, starting at different nodes of a network, possibly at different times, have to meet.
This basic task,  known  as {\em rendezvous}, has been thoroughly studied in the literature.
It even has applications in human and animal interaction, e.g., when agents are people that have to meet in a city whose streets form a network,
or migratory birds have to gather at one destination flying in from different places.
In computer science applications,
mobile agents usually represent software agents in computer networks, or mobile robots, if the network is a labyrinth or is composed of corridors in a building.
The reason to meet may be to exchange data previously collected by the agents,
or to coordinate some future task, such as network maintenance or finding a map of the network.

In this paper we consider the rendezvous problem under a very weak scenario which assumes little knowledge and control power of the agents.
This makes our solutions more widely applicable, but significantly increases the difficulty of meeting. More specifically, agents do not have any a priori information about the network, they do not know its topology or any bounds on parameters such as the diameter or the size. 
We seek rendezvous algorithms that do not
rely on the knowledge of node labels, and can work in anonymous networks as well  (cf. \cite{alpern02b}). 
The importance of designing such algorithms
is motivated by the fact that, even when nodes are equipped with distinct labels, agents may be unable to perceive them
because of limited sensory capabilities, 
or nodes may refuse to reveal their labels, e.g., due to security or privacy reasons.
Note that if nodes had distinct labels that can be perceived by the agents, then agents might explore the network and meet in the smallest node, hence rendezvous would reduce to exploration.
Agents have distinct labels, which are positive integers and each agent knows its own label, but not the label of the other agent. 
The label of the agent is the only a priori initial input to its algorithm. During navigation agents gain knowledge of the visited part of the network:
when an agent enters a node, it learns the port number by which it enters and the degree of the node.
The main difficulty of the scenario is the asynchronous way in which agents move: the speed of the agents may vary, may be different for each of them, and is totally controlled by an adversary.  This feature of the model is also what makes it more realistic than the synchronous scenario: in practical applications the speed of agents depends on various factors
that are beyond their control, such as congestion in different parts of the network or mechanical characteristics in the case of mobile robots.
Notice that in the asynchronous scenario we cannot require that agents meet in a node: the adversary can prevent this even in the two-node graph.
Thus, similarly as in previous papers on asynchronous rendezvous \cite{BCGIL,CCGL,CLP,DGKKPV,GP}, we allow the meeting either in
a node or inside an edge. 
The cost of a rendezvous algorithm is the total number of edge traversals by both agents until their
meeting.

\noindent
{\bf Our results.} The main result of this paper is a deterministic rendezvous algorithm, working in arbitrary unknown networks and whose cost is polynomial in the size of the network and in the {\em length} of the {\em smaller} label (i.e. in the logarithm of this label).
The only previous algorithm solving the asynchronous rendezvous problem \cite{CLP} is exponential in the size of the network and in the larger label.
Hence we decrease the cost exponentially in the size of the network and doubly exponentially in the labels of agents.

As an application of our rendezvous algorithm we solve several fundamental problems involving teams of unknown size larger than 1 of labeled agents
moving asynchronously in unknown networks. Among them are the following problems: {\tt team size}, in which every agent has to find the total
number of agents, {\tt leader election}, in which all agents have to output the label of a single agent, {\tt perfect renaming} in which all agents
have to adopt new different labels from the set $\{1,\dots,k\}$, where $k$ is the number of agents, and {\tt gossiping}, in which each agent has initially a piece of information (value) and all agents have to output all the values. Using our rendezvous algorithm we solve all these problems at cost
(total number of edge traversals by all agents) polynomial in the size of the graph and in the smallest length of all labels of participating agents.
To the best of our knowledge this is the first solution of these problems for asynchronous mobile agents, even regardless of the cost.

\noindent
{\bf The model.}
The network is modeled as a finite simple undirected connected graph (without self-loops or multiple edges), referred to hereafter as a graph. 
Nodes are unlabeled, but
edges incident to a node $v$ have distinct labels in 
$\{0,\dots,d-1\}$, where $d$ is the degree of $v$. Thus every undirected
edge $\{u,v\}$ has two labels, which are called its {\em port numbers} at $u$
and at $v$. Port numbering is {\em local}, i.e., there is no relation between
port numbers at $u$ and at $v$. Note that in the absence of port numbers, edges incident to a node
would be undistinguishable for agents and thus gathering would be often impossible, 
as the adversary could prevent an agent from taking some edge incident to the current node. 
{By $succ(v,i)$ we denote the neighbor of $v$ linked to it by the edge with port number $i$ at $v$.}

In order to avoid crossings of non-incident edges, we consider an embedding of the underlying graph in the
three-dimensional Euclidean space, with nodes of the graph being points of the
space and edges being
pairwise disjoint line segments joining them.
 Agents are modeled
as points moving inside this embedding. (This embedding is only for the clarity of presentation; in fact crossings of non-incident edges
would make rendezvous simpler, as agents traversing distinct edges could sometimes meet accidentally at the crossing point.)

There are two agents that start from different nodes of the graph  and  traverse its edges.
They cannot mark visited nodes or traversed edges in any way.
Agents have distinct labels which are strictly positive integers. Each agent knows only its own label which is an initial input
to its deterministic algorithm.
Agents do not  know the topology of the graph or any bound on its size. They can, however, acquire knowledge about the network:
When an agent enters a node, it learns its degree and the port of entry. We assume that the memory of the agents is unbounded: from the computational point of view they are modeled as 
Turing machines. 

Agents navigate in the graph in an asynchronous way which is formalized by an adversarial model used in 
\cite{BCGIL,CCGL,CLP,DGKKPV,GP} and described below.
Two important notions used to specify movements of agents are the {\em route} of the agent and its {\em walk}.
Intuitively, the agent chooses the route {\em where} it moves and the adversary describes the walk on this 
route, deciding {\em how} the agent  moves. More precisely,  these notions are defined as follows.
The adversary initially places an agent at some node of the graph.
The route is chosen by the agent and is defined as follows. 
The agent chooses one of the available ports at the current node. 
After getting to the other end of the corresponding edge, the agent chooses one of the available ports at this node
or decides to stay at this node. It does so on the basis of all information currently available to it.
The resulting route of the agent is the corresponding sequence of edges $(\{v_0,v_1\},\{v_1,v_2\},\dots)$,
which is a (not necessarily simple) path in the graph. 

We now describe the walk $f$ of an agent on its route. Let $R=(e_1,e_2,\dots)$ be the route of an agent. Let $e_i=\{v_{i-1},v_i\}$.
Let $(t_0,t_1,t_2,\dots)$, where $t_0=0$, be an increasing sequence of reals, chosen by the adversary, 
that represent points in time. Let $f_i:[t_i,t_{i+1}]\rightarrow [v_i,v_{i+1}]$ be any continuous function, chosen by the adversary, such that $f_i(t_i)=v_i$ and $f_i(t_{i+1})=v_{i+1}$. For any $t\in [t_i,t_{i+1}]$, we define $f(t)=f_i(t)$. 
The interpretation of the walk $f$ is as follows: at time $t$ the agent
is at the point $f(t)$ of its route.  This general definition of the walk and the fact that (as opposed to the route) it is designed by the adversary,
are a way to formalize the asynchronous characteristics of the process.  The movement of the agent can be
at arbitrary speed, the adversary may sometimes stop the agent or move it back and forth, as long as the walk 
in each edge of the route is continuous and covers all of it.
This definition makes the adversary very powerful,
and consequently agents have little control on how they move. This makes a meeting between agents hard to achieve.
Agents with routes $R_1$ and $R_2$ and with walks $f_1$ and $f_2$ meet at time $t$,
if points $f_1(t )$ and $f_2(t )$ are identical. 
A meeting is guaranteed for routes $R_1$ and $R_2$,
if the agents using these routes meet at some time $t$, regardless of the walks chosen by the adversary.

\noindent{\bf Related work.}
In most papers on rendezvous a synchronous scenario was assumed, in which agents navigate in the graph in synchronous rounds.
An extensive survey of  randomized rendezvous in various scenarios can be found in
\cite{alpern02b}, cf. also  \cite{alpern95a,alpern02a,alpern99,anderson90}. 
Deterministic rendezvous in networks has been surveyed in \cite{Pe}.
Several authors
considered the geometric scenario (rendezvous in an interval of the real line, see, e.g.,  \cite{baston01},
or in the plane, see, e.g., \cite{anderson98a,anderson98b}).
Rendezvous of  more than two agents, often called gathering, has been studied, e.g., 
in \cite{DiPe,DPP,lim96,YY}. In \cite{DiPe} agents were anonymous, while in~\cite{YY} the authors considered 
gathering many agents with unique labels. Gathering many labeled agents in the presence of Byzantine agents was studied in \cite{DPP}. 
The problem was also studied in the context of multiple robot systems, cf.
\cite{CP05,fpsw}, and fault tolerant gathering of robots in the plane was studied, e.g., in \cite{AP06,CP08}. 

For the deterministic setting a lot of effort has been dedicated to the study of the feasibility of rendezvous, and to the time required to achieve this task, when feasible. For instance, deterministic rendezvous with agents equipped with tokens used to mark nodes was considered, e.g., in~\cite{KKSS}. Deterministic rendezvous of two agents that cannot mark nodes but have unique labels was discussed in {\cite{DFKP,KM,TSZ14}}.
These papers are concerned with the time of synchronous rendezvous in arbitrary
graphs. In \cite{DFKP} the authors show a rendezvous algorithm polynomial in the size of the graph, in the length of the shorter
label and in the delay between the starting time of the agents. In {\cite{KM,TSZ14}} rendezvous time is polynomial in the first two of these parameters and independent of the delay.

Memory required by two anonymous agents to achieve deterministic rendezvous has been studied in \cite{FP2} for trees and in  \cite{CKP} for general graphs.
Memory needed for randomized rendezvous in the ring is discussed, e.g., {in~\cite{KKPM11}}.

Asynchronous rendezvous of two agents in a network has been studied in \cite{BCGIL,CCGL,CLP,DGKKPV,GP}. The model used in the present paper has been
introduced in \cite{DGKKPV}. In this paper  the authors investigated the cost of rendezvous in the infinite line and in the ring. They also proposed a rendezvous
algorithm for an arbitrary graph with a known upper bound on the size of the graph. This assumption was subsequently removed in \cite{CLP}, but both in
 \cite{DGKKPV} and in  \cite{CLP} the cost of rendezvous was exponential in the size of the graph and in the larger label. In \cite{GP} asynchronous rendezvous was studied for anonymous agents and the cost was again exponential. The only asynchronous rendezvous algorithms at polynomial cost were
 presented in \cite{BCGIL,CCGL}, but in these papers authors restricted attention to infinite multidimensional grids and they used 
 the powerful assumption that each agent knows its starting coordinates. (The cost in this case is polynomial in the initial distance).
 
 A different asynchronous scenario was studied in \cite{CFPS,fpsw} for the plane. In these papers the authors assumed that agents are memoryless, but they can observe the environment and make navigation decisions based on these observations.
 
 The four problems that we solve in the context of asynchronous mobile agents as an application of our rendezvous algorithm, are
 widely researched tasks in distributed computing, under many scenarios. Counting the number of agents is a basic task, cf. \cite{FP3},
 as many mobile agents algorithms depend on this knowledge.  Leader election, cf. \cite{Ly},  is a fundamental problem in distributed
 computing. Renaming  was introduced in \cite{ABDKPR} and further studied by many authors. Gossiping, also called all-to-all communication,
 is one of the basic primitives in network algorithms, cf. \cite{FL}.

\section{Preliminaries}\label{prelim}

Throughout the paper, 
the number of nodes of a graph is called its size.
In this section we {present two procedures}, that will be used as building blocks in our algorithms. 
The aim of both of them is graph exploration, i.e., visiting all nodes and traversing all edges of the graph by a single agent. 
The first procedure, based on universal exploration sequences (UXS), is a corollary of the  result of Reingold \cite{Re}. Given any positive integer $n$, it allows the agent to traverse all edges of any graph of size at most $n$,
starting from any node of this graph, using $P(n)$ edge traversals, where $P$ is some polynomial. (The original procedure of Reingold only visits all nodes, but it can be transformed to traverse all edges by visiting all neighbors of each visited node before going to the next node.) After entering a node of degree $d$ by some port $p$,
the agent can compute the port $q$ by which it has to exit; more precisely $q=(p+x_i)\mod d$, where $x_i$ is the corresponding term of the UXS.

A {\em trajectory} is a sequence of nodes of a graph, in which each node is adjacent to the preceding one. (Hence it is a sequence of nodes visited following a route.)
Given any starting node $v$,  we denote by $R(n,v)$ the trajectory obtained by Reingold's procedure. The procedure can be applied in any graph starting at any node, giving
some trajectory. We say that  the agent {\em follows} a trajectory if it executes the above procedure used to construct it.
This trajectory will be called {\em integral}, if the corresponding route covers all edges of the graph. By definition, the trajectory $R(n,v)$ is integral if it is
obtained by Reingold's procedure applied in any graph of size at most $n$ starting at any node~$v$. 

{The second procedure, derived from \cite{DP} and adapted here to our needs, allows an agent to traverse all edges {and visit all nodes} of any graph of size at most $n$ provided that there is a {unique token located on an {\em extended edge} $u-v$ (for some adjacent nodes $u$ and $v$) and authorized to move arbitrarily on it.
An extended edge is defined as the edge $u-v$ augmented by nodes $u$ and $v$}. (It is well known that a terminating exploration even of all anonymous rings of unknown size by a single agent without a token is impossible.) 
In our applications the roles of the token and of the exploring agent will be played by agents. We call this procedure $ESST$, for {\em exploration with a semi-stationary token}, {as the token always remains on the same extended edge (even if it can move arbitrarily inside it)}. We first describe the procedure before showing its validity as well as its polynomial complexity with respect to the size of the graph.}

{
The following notion will be crucial for our considerations. Let $m$ be a positive integer.  An application of $R(2m,u)$ to a graph $G$ at some node $u$ is called {\em clean}, if all nodes in this application are of degree at most $m-1$.}

\vspace*{0.2cm}
\noindent
{{\bf Procedure ESST}}

{The algorithm proceeds in {phases $i=3,6,9,12\dots$.} In any phase $i$, the agent first applies 
$R(2i,v)$
at the node $v$ in which it started this phase.
Let $(u_1,\dots , u_{r+1})$ be the trajectory $R(2i,v)$ ($v=u_1$ and $r=P(2i)$). Call this trajectory the trunc of this phase.
If it is not clean, or if no token is seen, the agent aborts phase $i$ and starts phase {$i+3$}. Otherwise, the agent backtracks to $u_1$, and applies 
$R(i,u_j)$
at each node $u_j$ of the trunc, interrupting a given execution of $R(i,u_j)$ when it sees a token, every time recording the {\em code} of the path from $u_j$ to this token. 
This code is defined as the sequence of ports
encountered while walking along the path.
(If, for some $j$, the token is at $u_j$, then this code is an empty sequence.) After seeing a token, the agent backtracks to $u_j$, goes
to $u_{j+1}$ and starts executing $R(i,u_{j+1})$. For each node $u_j$ of the trunc, if at the end of $R(i,u_j)$ either no token is seen during the execution of $R(i,u_j)$, or the agent has recorded at least {$\frac{i}{3}$} different codes in phase $i$, then
the agent aborts phase $i$ and starts phase {$i+3$} (in the special case where the agent decides to abort phase $i$ while traversing an edge, phase {$i+3$} starts at the end of this edge traversal). Otherwise, upon completion of phase $i$, it stops as soon as it is at a node.}

\vspace*{0.2cm}

{The remaining part of this section is devoted to the proof that Procedure ESST is correct and works at polynomial cost. Again, this result uses  ideas from \cite{DP}: we include it for the sake of completeness.}

\begin{lemma}\label{lemma}
{Let $m\leq n$ be positive integers, and let $G$ be a graph of size $n$. Let $S$ be the trajectory in $G$ resulting from the execution of $R(2m,v)$, for some node $v$ of $G$. If trajectory $S$ is clean, then $S$ contains at least $m$ different nodes.}
\end{lemma}

\begin{proof}
{Let $S=(u_1,\dots , u_{r+1})$ with $v=u_1$ and $r=P(2m)$.
Suppose for contradiction that there are fewer than $m$ different nodes in $S$, and let $X$ be the set of these nodes. Consider any node $x \in X$.
A port $j$ at node $x$ is called {\em occupied}, if for some index $t$, we have $x=u_t$ and either $succ(u_t,j)=u_{t-1}$ or   $succ(u_t,j)=u_{t+1}$.
Otherwise it is called {\em free}.
Let $d$ be the maximum number of free ports at any node of $X$.
Construct the following graph $H$. The set of nodes of $H$ is $X\cup\{y_1,\dots,y_d\}$, where all $y_s$ are distinct and do not belong to $X$. 
The set of edges of the graph $H$ consists of all edges $\{u_t,u_{t+1}\}$ from $G$ augmented by the following set of edges. Consider all nodes $x\in X$
in the order of their first appearance in the sequence $S$. Let $c_1,\dots,c_p$ be the free ports at $x$, listed in increasing order. We add edges $\{x,y_1\},\dots, \{x,y_p\}$ 
with the following ports: the port at $x$ corresponding to the edge  $\{x,y_q\}$ is $c_q$, and the port at $y_q$ corresponding to the edge  $\{x,y_q\}$ is the smallest port not yet used
at this node. This completes the construction of graph~$H$.}

{Since trajectory $S$ is clean, we have $d<m$. Since the size of $X$ is smaller than $m$, the graph $H$ has fewer than $2m$ nodes.
Since the size of $X$ is smaller than $n$ (in view of $m \leq n$), at least one port at some node of $X$ is free, and consequently $d\geq 1$.
It follows that some nodes $y_s$ were added to $G$ in order to construct $H$.
Nodes $y_s$ are not terms of the trajectory $S$ in $H$. This is a contradiction with the fact that $R(2m,x)$ allows to visit all nodes of any graph $F$ of size at most $2m$ from any node $x$ of $F$.}
\end{proof}

{We are now ready to prove the following theorem.}

\begin{theorem}
\label{theo:est}
{Procedure $ESST$ terminates in every graph $G$ after a number of steps polynomial in the size of $G$. Upon its termination, all {edges of $G$ are traversed 
by the agent}.}
\end{theorem}

\begin{proof}
{Let $n$ be the size of the graph $G$.
First observe that the procedure terminates at the latest after completion of phase {$i=9n+3$}. Indeed in this phase, every trajectory $R(2i,v)$, for any node $v$ of $G$, must be clean.
{Moreover, by the end of each trajectory $R(i,u)$, for any node $u$ of the trunc, the token must be met. Finally,} the number of possible codes recorded by the agent cannot exceed {$3n$ (there are at most $3n$ different codes in a graph of size $n$, as there are at most $3$ different codes for each node, {depending on whether the token is at one of the nodes or inside the edge corresponding to the extended edge on which it is located)} and thus at most $\frac{i-3}{3}<\frac{i}{3}$ codes in phase $i$.}}

{Let us estimate the number of edge traversals executed by the agent in some phase $j\geq1$ from some node $v$. The agent walks at most three times along the trunc corresponding to the trajectory $R(2j,v)$, and at most twice along each trajectory $R(j,u)$ from each node $u$ of the trunc of phase $j$. This gives a total of at most $3P(2j)+P(2j)*P(j)$ edge traversals, which is polynomial in $j$.}
{Hence the total number of edge traversals made by the agent by the end of phase $i$ is upper bounded by {$\frac{i}{3}.(3P(2i)+P(2i)*P(i))=\frac{9n+3}{3}(3P(2(9n+3))+P(2(9n+3))*P(9n+3))$, which is polynomial in $n$.}}

{It remains to show that if the agent stops upon completion of some phase $t \leq i$, then all {edges are traversed}. Consider this phase $t$.  By the description of the procedure, the main trajectory, corresponding to the trunc of phase $t$, must be clean, the token must be seen
in each trajectory $R(t,u)$, for each node $u$ of the trunc, made in this phase, and the number of codes recorded by the agent cannot exceed {$\frac{t}{3}-1$}.
Suppose by contradiction that $n\geq t$.
By Lemma \ref{lemma}, the set of distinct nodes visited during the main trajectory $R(2t,v)$ in phase $t$ is at least $t$.
On the other hand, by the description of the procedure, the agent has recorded at most {$\frac{t}{3}-1$} different codes.}
{Hence, there are $3$ distinct nodes $u'$, $u''$ and $u'''$ visited during the main trajectory $R(2t,v)$ and at least one code $C$ such that $C$ was recorded from each of these nodes (among  other possible codes from these nodes).}

{However, a given code cannot be recorded from more than two distinct nodes as otherwise that would imply that there is more than one token in $G$ or there is a node $w$ in the graph, for which edges to two of its neighbors correspond to the same port number at $w$, which is impossible. Hence, $u'$, $u''$ and $u'''$ cannot be all distinct, which is a contradiction.}

{This implies that $n<t$, and consequently all {edges of $G$ are traversed} during the main trajectory $R(2t,v)$. It follows that upon completion of phase $t$ all {edges of $G$ are traversed}.}
\end{proof}

{Note that the fact that all the edges of a graph $G$ are traversed during an execution of procedure $ESST$ implies that all nodes of $G$ are visited.}
We denote by $T(ESST(n))$ the maximum number of edge traversals in an execution of the procedure $ESST$
in a graph of size at most $n$.

{To complete this section, let us introduce some more notation.} For a positive integer $x$, by $|x|$ we denote the length of its binary representation, called the length of $x$. Hence $|x|=\lceil \log x \rceil$. All logarithms are with base 2.
For two agents, we say that the agent with larger (smaller) label is larger (resp. smaller). For any trajectory $T=(v_0,\dots,v_r)$, we denote by $\overline{T}$  the reverse trajectory $(v_r,\dots,v_0)$. For two trajectories $T_1= (v_0,\dots,v_r)$ and $T_2=(v_r,v_{r+1},\dots ,v_s)$ we denote by $T_1T_2$ the trajectory 
$(v_0,\dots,v_r, v_{r+1},\dots ,v_s)$. For any trajectory $T=(v_0,\dots,v_r)$, for which $v_r=v_0$ and for any positive integer $x$, we define $T^x$ to be $TT\dots T$, 
with $x$ copies of $T$.
For any trajectory $T$ we define $|T|$ to be the number of {edge traversals} in $T$.

\section{The rendezvous algorithm}

In this section we describe and analyze our rendezvous algorithm working at polynomial cost. Its high-level idea is based on the following observation.
If one agent follows an integral trajectory during some time interval, then it must either meet the other agent or this other agent must perform at least one complete edge traversal during this time
interval, i.e., it must make {\em progress}.  
A naive use of this observation leads to the following simple algorithm: an agent with label $L$ starting at node $v$ of a graph of size $n$ follows the trajectory $(R(n,v)\overline{R(n,v)})^{(2P(n)+1)^L}$
and stops.
Indeed, in this case the number of integral trajectories $R(n,v)\overline{R(n,v)}$ performed by the larger agent is larger than the number of edges traversed by the smaller agent and consequently, if they have not met before,  the larger agent must meet the smaller one after the smaller agent stops, because the larger agent will still perform at least one entire trajectory afterwards.
However,  this simple algorithm has two major drawbacks.
First, it requires knowledge of $n$ (or of an upper bound on it) and second, it is exponential in $L$, while we want an algorithm {\em polylogarithmic} in $L$. Hence the above observation has to be used in a much more subtle way. Our algorithm constructs a trajectory for each agent, polynomial in the size of the graph and polylogarithmic in the shorter label, i.e., polynomial in its length, which has the following {\em synchronization} property that holds in a graph of arbitrary unknown size. 
When one of the agents has already followed some part of its trajectory, it has either met the other agent, or this other agent must have completed
some other related part of its trajectory. (In a way, if the meeting has not yet occurred, the other agent has been ``pushed'' to execute some part of its route.) The trajectories are designed in such a way
that, unless a meeting has already occurred, the agents are forced to follow in the same time interval such parts of their trajectories that meeting is inevitable. A design satisfying this
synchronization property is difficult due to the arbitrary behavior of the adversary and is the main technical challenge of the paper.

\subsection{Formulation of the algorithm}

We first define several trajectories based on trajectories $R(k,v)$. Each trajectory
is defined using a starting node $v$ and a parameter $k$. Notice that, similarly as the basic trajectory $R(k,v)$, each of these trajectories (of increasing complexity)
can be defined in any graph, starting from any node~$v$: {in particular, for a fixed parameter $k$, a given trajectory always traverses the same number of edges, regardless of the graph and of the starting node}.

\begin{definition}
The trajectory $X(k,v)$ is the sequence of nodes  $R(k,v)\overline{R(k,v)}$.
\end{definition}

{In other terms, trajectory $X(k,v)$ consists in following trajectory $R(k,v)$ and then backtracking to node $v$ by following the reverse path $\overline{R(k,v)}$}.

\begin{definition}
\label{def:Q}
The trajectory $Q(k,v)$ is the sequence of nodes $X(1,v)X(2,v)\dots X(k,v)$ {(refer to Figure~\ref{fig:Q})}.
\end{definition}

\begin{figure*}[httb!]
	\begin{center}
	\includegraphics[width=0.3\textwidth]{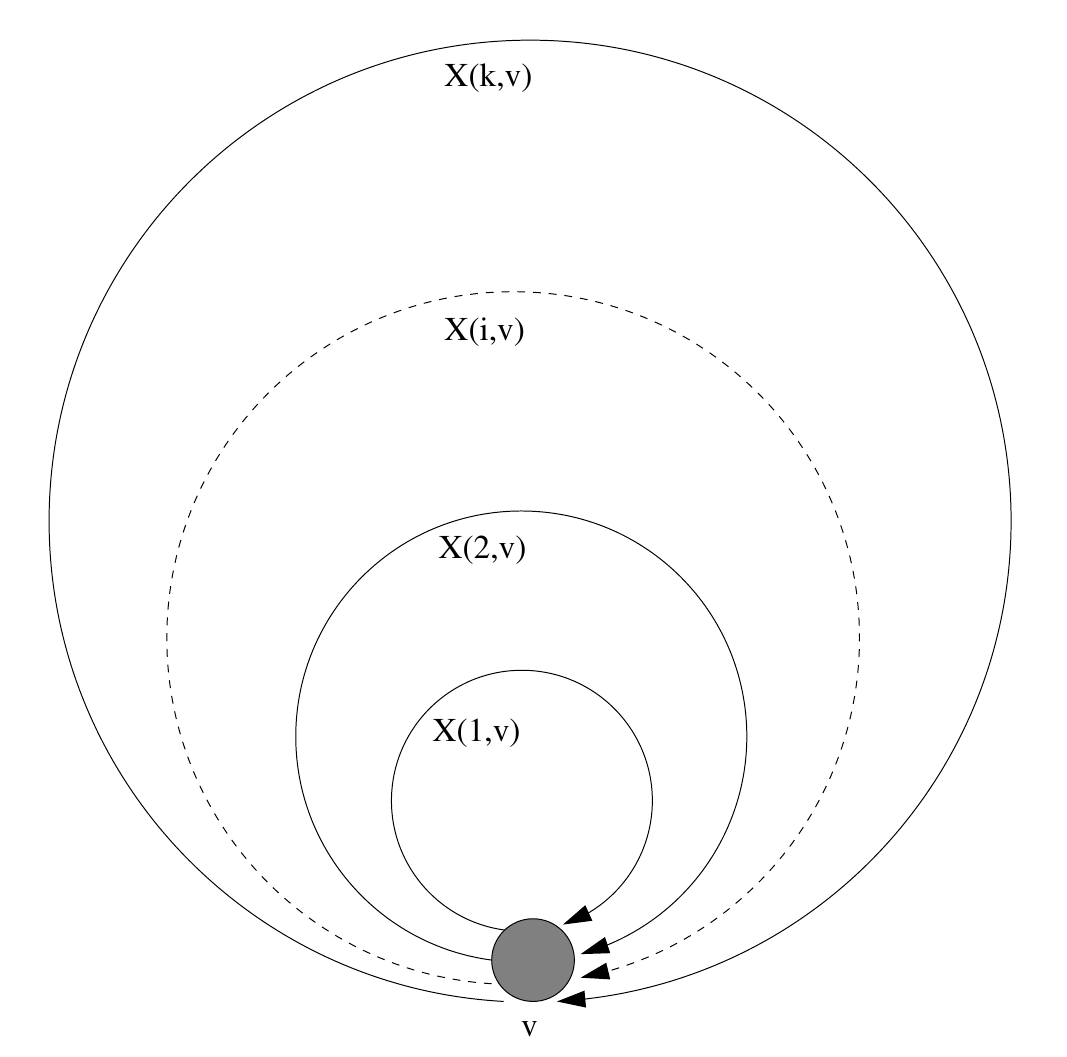}
	\caption{{A schematic representation of trajectory $Q(k,v)$ which is made up of a sequence of consecutive trajectories $X(i,v)$ from $i=1$ to $i=k$.}}
	\label{fig:Q}
	\end{center}
\end{figure*}

\begin{definition}
\label{def:Y}
Let $R(k,v_1)=(v_1,v_2,\dots v_s)$. Let $$Y'(k,v_1)=Q(k,v_1)(v_1,v_2)Q(k,v_2)(v_2,v_3)Q(k,v_3)\dots (v_{s-1},v_s) Q(k,v_s).$$ We define
the trajectory $Y(k,v_1)$ as $Y'(k,v_1)\overline{Y'(k,v_1)}$. 
\end{definition}

{In other terms, trajectory $Y(k,v_1)$ consists in following trajectory $Y'(k,v_1)$ and then backtracking to node $v_1$ by following the reverse path $\overline{Y'(k,v_1)}$. A schematic representation of $Y'(k,v_1)$ is depicted in Figure~\ref{fig:Y'}.}

\begin{figure*}[httb!]
	\begin{center}
	\includegraphics[width=0.5\textwidth]{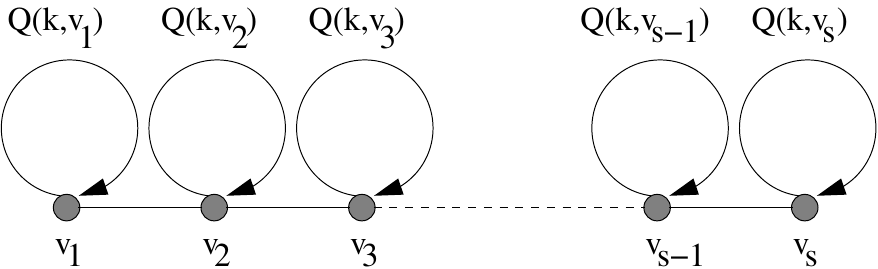}
	\caption{{A schematic representation of trajectory $Y'(k,v_1)$ which consists in following trajectory $R(k,v_1)=(v_1,v_2,\dots v_s)$ with the following insertions: for all $i<s$, before going from node $v_i$ to $v_{i+1}$ the agent follows trajectory $Q(k,v_i)$.}}
	\label{fig:Y'}
	\end{center}
\end{figure*}

\begin{definition}
\label{def:Z}
The trajectory $Z(k,v)$ is the sequence of nodes $Y(1,v)Y(2,v)\dots Y(k,v)$ {(refer to Figure~\ref{fig:Z})}.
\end{definition}

\begin{figure*}[httb!]
	\begin{center}
	\includegraphics[width=0.3\textwidth]{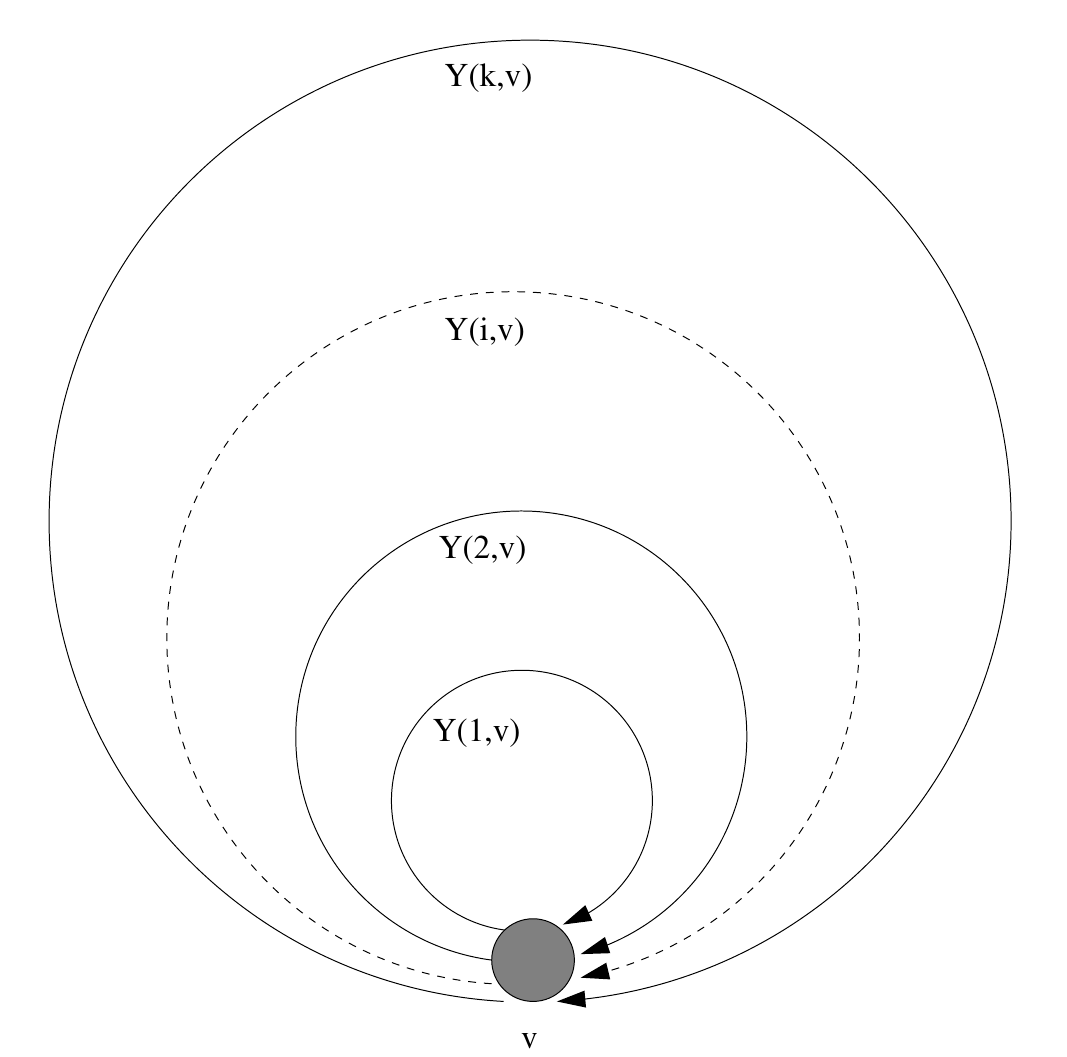}
	\caption{{A schematic representation of trajectory $Z(k,v)$ which is made up of a sequence of consecutive trajectories $Y(i,v)$ from $i=1$ to $i=k$.}}
	\label{fig:Z}
	\end{center}
\end{figure*}

\begin{definition}
\label{def:A}
Let $R(k,v_1)=(v_1,v_2,\dots v_s)$. Let $$A'(k,v_1)=Z(k,v_1)(v_1,v_2)Z(k,v_2)(v_2,v_3)Z(k,v_3)\dots (v_{s-1},v_s) Z(k,v_s).$$ We define
the trajectory $A(k,v_1)$ as $A'(k,v_1)\overline{A'(k,v_1)}$. 
\end{definition}
\newpage
{In other terms, trajectory $A(k,v_1)$ consists in following trajectory $A'(k,v_1)$ and then backtracking to node $v_1$ by following the reverse path $\overline{A'(k,v_1)}$. A schematic representation of $A'(k,v_1)$ is depicted in Figure~\ref{fig:A'}.}

\begin{figure*}[httb!]
	\begin{center}
	\includegraphics[width=0.5\textwidth]{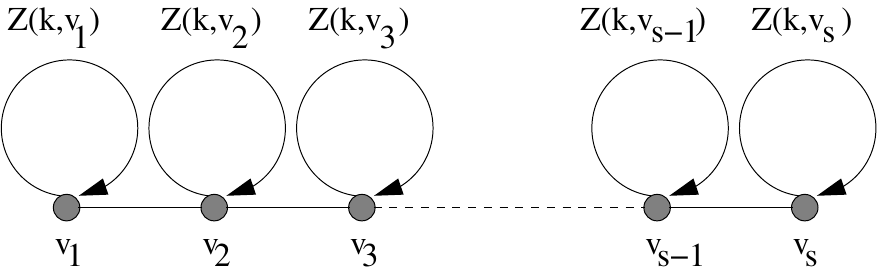}
	\caption{{A schematic representation of trajectory $A'(k,v_1)$ which consists in following trajectory $R(k,v_1)=(v_1,v_2,\dots v_s)$ with the following insertions: for all $i<s$, before going from node $v_i$ to $v_{i+1}$ the agent follows trajectory $Z(k,v_i)$.}}
	\label{fig:A'}
	\end{center}
\end{figure*}

{If the node $v$ is clear from the context, we will sometimes omit it, thus writing $X(k)$ instead of $X(k,v)$, etc.}

{In the following definition, $|A(4k)|$ corresponds to the number of edges that are traversed by following trajectory $A(4k,s)$ for any node $s$.} 
\begin{definition}
\label{def:B}
The trajectory $B(k,v)$ is the sequence of nodes $Y(k,v)^{2|A(4k)|}$. 
\end{definition}

{Below is the pseudocode describing how to follow trajectory $B(k,v)$.}

\begin{center}
\fbox{
\begin{minipage}{0.4\columnwidth}\small
{\bf for} $i$ {\bf from} 1 {\bf to} $2|A(4k)|$\\
\hspace*{0.8cm}Follow trajectory $Y(k,v)$\\
{\bf end for}
\end{minipage}
}
\end{center} 

{In the following definition, the value $|B(4k)|$ (resp. $|A(8k)|$) corresponds to the number of edges that are traversed by following trajectory $B(4k,s)$ (resp. $A(8k,s)$) for any node $s$.}

\begin{definition}
\label{def:K}
The trajectory $K(k,v)$ is the sequence of nodes $X(k,v)^{2(|B(4k)|+|A(8k)|)}$. 
\end{definition}

{Below is the pseudocode describing how to follow trajectory $K(k,v)$.}

\begin{center}
\fbox{
\begin{minipage}{0.4\columnwidth}\small
{\bf for} $i$ {\bf from} 1 {\bf to} $2(|B(4k)|+|A(8k)|)$\\
\hspace*{0.8cm}Follow trajectory $X(k,v)$\\
{\bf end for}
\end{minipage}
}
\end{center} 

{In the following definition, the value $|K(k)|$ corresponds to the number of edges that are traversed by following trajectory $K(k,s)$ for any node $s$.}

\begin{definition}
\label{def:ome}
The trajectory $\Omega(k,v)$ is the sequence of nodes $X(k,v)^{(2k-1)|K(k)|}$.
\end{definition}

{Below is the pseudocode describing how to follow trajectory $\Omega(k,v)$.}

\begin{center}
\fbox{
\begin{minipage}{0.4\columnwidth}\small
{\bf for} $i$ {\bf from} 1 {\bf to} $(2k-1)|K(k)|$\\
\hspace*{0.8cm}Follow trajectory $\Omega(k,v)$\\
{\bf end for}
\end{minipage}
}
\end{center}

{Using the above defined trajectories we are now ready to describe Algorithm RV-asynch-poly executed by an agent with label $L$ in an arbitrary graph. Below we give the pseudocode of this algorithm and then the main intuition that is behind it.}

The agent first modifies its label. If $x=(c_1\dots c_r)$ is the binary representation of $L$, define the {\em modified label} of the agent to be the sequence $M(x)=(c_1c_1c_2c_2\dots c_rc_r01)$.  
Note that, for any $x$ and $y$, the sequence $M(x)$ is never a prefix of $M(y)$.
Also, $M(x) \neq M(y)$ for $x\neq y$. 

\vspace*{0.5cm}

\begin{center}
\fbox{
\begin{minipage}{11cm}

\noindent
{\bf Algorithm RV-asynch-poly}.

\vspace*{0.5cm}

Let $x$ be the binary representation of the label $L$ of the agent and let $M(x)=(b_1b_2\dots b_s)$. Let $v$ be the starting node of the agent.

\vspace*{0.5cm}

\noindent
Execute until rendezvous.

\noindent
$i=1$;\\
$k=1$;\\
{\bf repeat}\\
\hspace*{0.5cm}{\bf  while}  $i\leq min(k,s)$ {\bf do} \\
\hspace*{1cm}{\bf if} $b_i=1$ {\bf then} follow the trajectory $B(2k,v)^2$\\
\hspace*{1cm}{\bf else} follow the trajectory $A(4k,v)^2$\\
\hspace*{1cm}{\bf if} $min(k,s)> i$ {\bf then} follow the trajectory $K(k,v)$\\
\hspace*{1cm}{\bf else} follow the trajectory $\Omega(k,v)$\\
\hspace*{1cm}$i:=i+1$\\
\hspace*{0.5cm}$i:=1$\\
\hspace*{0.5cm}$k:=k+1$

\end{minipage}
}
\end{center}

{The main idea of the above formulated algorithm is the following. In order to guarantee rendezvous, symmetry in the actions of the agents must be broken.
 Since agents have different transformed labels, this can be done by designing the algorithm so that each agent processes consecutive bits of its transformed label, acting differently when the current bit is 0 and when it is 1.  (The way of processing each bit is described in the ``while'' loop.) The aim is to force rendezvous when each agent processes the bit corresponding to the position where their transformed labels {first} differ. This approach requires overcoming two major difficulties. The first is that, due to the behavior of the asynchronous adversary,
 agents may execute corresponding bits of their transformed labels at different times. This problem is solved in our algorithm by using trajectories of type $K$ and $\Omega$ (refer to Definitions~\ref{def:K} and~\ref{def:ome}), in order to synchronize the agents. These trajectories have the following role in this synchronization effort: for $k\geq n$, trajectories $K(k)$ and $\Omega(k)$ executed by one agent push the other agent to proceed in its execution or otherwise rendezvous is accomplished. The joint application of these two specific trajectories {ends up forcing} the agents to push each other in such a way that at some point they process almost simultaneously the bit on which they differ.}

{ The second difficulty is to orchestrate rendezvous after the first difficulty has been overcome, i.e., when each agent 
 processes this bit. This is done by making use of trajectories of type $A$ and $B$ (refer to Definitions~\ref{def:A} and~\ref{def:B}). Our algorithm is designed in such a way that processing bit $0$ consists in following twice a trajectory of type $A$ (for some parameters), while processing bit $1$ consists in following twice a trajectory of type $B$ (for some parameters). This choice in the design stems from the desire to exploit the following feature: if two agents $a$ and $b$ simultaneously start to follow respectively trajectory $A(k,u)$ and $B(k,v)$ (for any $k\geq n$ and for any nodes $u$ and $v$) the rendezvous must occur by the time an agent terminates its trajectory first. Indeed, this is the case if $A(k,u)$ is finished before $B(k,v)$, as $B(k,v)$ consists in repeating $Y(k,v)= Y'(k,v)\overline{Y'(k,v)}$, while $A(k,u)$ allows agent $a$ to follow $Y'(k,s)\overline{Y'(k,s)}$ at least once from every node $s$ of the graph: roughly speaking, agent $a$ ends up ''catching'' agent $b$.
Otherwise (when $B(k,v)$ is finished by the time $A(k,u)$ is finished), this is also the case, as $B(k,v)$ consists in repeating the trajectory $Y(k,v)$, which is integral, more times than there are edges to traverse when following $A(k,u)$: roughly speaking, agent $b$ ends up ''catching'' agent $a$.}

{Of course, the occurrence of the kind of situation described above is ideal. However, we actually ensure only the occurence of a more general situation in which trajectories $A$ and $B$ may be followed from starting times that are "slightly" different and for a parameter $k$ that may also be different for each of them. Hence, to handle this, the algorithm is enriched by additionnal technical ingredients that are necessary to guarantee correctness. (It is particularly for these technical reasons that trajectory $A$ (resp. B) is repeated twice instead of only once when the processed bit corresponds to $0$ (resp. 1) and that agents follow trajectory $A(4k)$ or $B(2k)$ instead of simply $A(k)$ or $B(k)$).}

{We will show that the synchronization, and hence also rendezvous, always occurs soon enough to guarantee that every execution of the algorithm has necessarily a polynomial cost.}

\subsection{Proof of correctness and cost analysis}

We will use the following terminology refering to parts of the trajectory constructed by Algorithm RV-asynch-poly. The part before the start of $\Omega(1,v)$ is called the {\em first piece}
and is denoted $\mathcal{T}(1)$,
the part between the end of $\Omega(1,v)$ and the beginning of $\Omega(2,v)$ is called the {\em second piece} and is denoted $\mathcal{T}(2)$, etc. In general, the part  
 between the end of $\Omega(i-1,v)$ and the beginning of $\Omega(i,v)$ is called the $i$th {\em  piece} and is denoted $\mathcal{T}(i)$.
 The trajectory $\Omega(r,v)$ between pieces $\mathcal{T}(r)$ and $\mathcal{T}(r+1)$, is called the $r$th {\em fence}.
 
 Inside each piece, the trajectory $B(2k,v)^2$ and the trajectory  $A(4k,v)^2$ are called {\em segments}. Each of the two trajectories $B(2k,v)$ in the segment $B(2k,v)^2$ and each of the two trajectories  $A(4k,v)$ in the segment  $A(4k,v)^2$ are called {\em atoms}. We denote by $S_i(k)$ the segment in the $k$th piece corresponding to the bit $b_i$ in $M(x)$.  Each trajectory
 $K(k,v)$ is called a {\em border}. We denote by $K_{j,j+1}(k)$ the border between the segment $S_j(k)$ and the segment $S_{j+1}(k)$.
 
 We start with the following fact that will be often used in the sequel.
 
 \begin{lemma}\label{tunel}
 Suppose that agents $a$ and $b$ operate in a graph $G$. Let $v$ be a node of $G$ and let $m$ be a positive integer. If in some time interval $I$ agent $b$ keeps repeating the trajectory $X(m,v)$ and agent $a$ follows at least one entire trajectory $X(m,v)$, then the agents must meet during time interval $I$.
 The lemma remains true when $X$ is replaced by $Y$.
 \end{lemma}
 
 \begin{proof}
 Let $R=R(m,v)$. By definition, $X(m,v)=R\overline{R}$. During the time interval $I$ agent $a$ follows the entire trajectory $R$ at least once.
 If at the time when $a$ starts following $R$, agent $b$ is following $\overline{R}$, then they have to meet before $a$ finishes ${R}$ because $b$
is on a reverse path with respect to $a$. If at the time when $a$ starts following $R$, agent $b$ is also following ${R}$, then they are two cases to consider.

Case 1. ${a}$ completes trajectory $R$ before $b$ or simultaneously.\\ In this case they must meet because $a$ ``catches'' $b$.

Case 2.  $b$ completes trajectory $R$ before $a$.\\ In this case agent $b$ starts following trajectory $\overline{R}$ before  the time when $a$ completes $R$.
Hence agents must meet by the time $a$ completes trajectory $R$ because $b$ is on a reverse path with respect to $a$. 

For $Y$ instead of $X$ the argument is similar.
 \end{proof}

The following five lemmas establish various synchronization properties concerning the execution of the algorithm by the agents. 
They show that, unless agents have already met before, if one agent executes some part of  
Algorithm RV-asynch-poly, then the other agent must execute some other related part of it. These lemmas show the interplay of pieces, fences, segments,
atoms and borders that are followed by each of the agents:
these trajectories are the milestones of synchronization. In all lemmas we
suppose that agents $a$ and $b$ execute Algorithm RV-asynch-poly in a graph of size $n$, and we let $l$ to be the length of the smaller of their modified labels.  

\begin{lemma}
\label{lem:1}
If the agents have not met before, then by the time one of the agents completes its  $(n+l+i)$th fence, then the other agent must have completed its $(i+1)$th piece.
\end{lemma}

\begin{proof}
Without loss of generality assume that agent $b$ is the first to complete its $(n+l+i)$th fence 
$\Omega_b(n+l+i)$. When $b$ completed its $(n+l)$th fence $\Omega_b(n+l)$, agent $a$ must have completed its first piece $\mathcal{T}_a(1)$, otherwise $a$ and $b$ must have met because the trajectory $\Omega_b(n+l)$ contains more integral trajectories $X(n+l)$ than there are {edge traversals} in the trajectory $\mathcal{T}_a(1)$. Indeed, according to Algorithm RV-asynch-poly, the number of {edge traversals} in $\mathcal{T}_a(1)$ is bounded by $2(|A(4)|+|B(2)|)$, while according to Definitions~\ref{def:ome} and~\ref{def:K}, the number of integral trajectories $X(n+l)$ within $\Omega_b(n+l)$ is equal to {$(2(n+l)-1)|K(n+l)|=(2(n+l)-1)(|B(4(n+l))|+|A(8(n+l))|)|X(n+l)|$}, which is larger than $2(|A(4)|+|B(2)|)$ since $n+l\geq 2$.

When $b$ completes its $(n+l+1)$th piece $\mathcal{T}_b(n+l+1)$, agent $a$ must have completed its first fence $\Omega_a(1)$. Suppose not. This implies
that  while agent $b$ follows $\mathcal{T}_b(n+l+1)$, agent $a$ must follow
only its first fence $\Omega_a(1)$ or a part of it. This fence consists of repeating the trajectory
$X(1)$. Agent $b$ follows at some point the trajectory $A(4(n+l+1))$
or the trajectory $B(2(n+l+1))$ in its $(n+l+1)$th piece $\mathcal{T}_b(n+l+1)$.
By Definitions \ref{def:A} and \ref{def:B}, agent $b$ must have completed $X(1,u)$, for any node $u$ of the graph,
and hence must have met $a$, in view of Lemma \ref{tunel} which is a contradiction.

Similarly we prove that when $b$ completes its $(n+l+1)$th fence $\Omega_b(n+l+1)$, agent $a$ must have completed its second piece $\mathcal{T}_a(2)$,
and when $b$ completes its $(n+l+2)$th piece $\mathcal{T}_b(n+l+2)$, agent $a$ must have completed its second fence $\Omega_a(2)$. 
In general, it follows by induction on $i$ that when $b$ completes its $(n+l+i)$th fence $\Omega_b(n+l+i)$, agent $a$ must have completed its $(i+1)$th piece $\mathcal{T}_a(i+1)$.
\end{proof}

\begin{lemma}
\label{claim1}
Let $b$ be the first agent to complete its $(2(n+l))$th fence. If the agents have not met before, then during the time segment in which agent $b$ follows its
$(2(n+l))$th fence, agent $a$ follows a trajectory included in {$M\Omega_a(j)N$}, for some fixed $j$ satisfying $n+l+1\le j \le 2(n+l)$, where {$M$} is the last atom of its $j$th piece $\mathcal{T}_a(j)$, $\Omega_a(j)$ is its $j$th fence,
and {$N$} is the first atom of its $(j+1)$th piece $\mathcal{T}_a(j+1)$.
This $j$ will be called the index of agent $a$.
\end{lemma}

\begin{proof}
Consider the time interval $I$ during which agent $b$ follows its $(2(n+l))$th fence $\Omega_b(2(n+l))$. If during this time interval agent $a$ has not started any fence,
it would have to follow a trajectory included in a piece $\mathcal{T}_a(k)$ for some $k\le 2(n+l)$, because $b$ was the first agent to complete its $(2(n+l))$th fence.
By Definition~\ref{def:ome}, the trajectory $\Omega_b(2(n+l))$ contains more copies of the integral trajectory {$X(2(n+l))$} than there are edge traversals done by agent $a$. 
Indeed, according to Algorithm RV-asynch-poly, the number of {edge traversals} in $\mathcal{T}_a(k)$ is bounded by {$(k-1)|K(k)| + k(2(|A(4k)|+|B(2k)|)) < (2k-1)|K(k)|$}, which is at most $2(2(n+l)-1)|K(2(n+l))|$ for $k\le 2(n+l)$, while the number of integral trajectories $X(2(n+l))$ in $\Omega_b(2(n+l))$ is equal to $2(2(n+l)-1)|K(2(n+l))|$.
Hence the agents would have met, which is a contradiction.

Hence agent $a$ must have started some fence during the time interval $I$. By Lemma~\ref{lem:1}, during the time interval $I$ agent $a$ must have started its
$j$th fence  $\Omega_a(j)$, for some $j\in\{n+l+1,\ldots,2(n+l)\}$. Moreover, during the time interval $I$ agent $a$ could not have followed the entire last atom {$M$} 
of its $j$th piece $\mathcal{T}_a(j)$. Indeed, this would mean that during the time interval $I$ agent $a$ has 
entirely followed either the trajectory $B(2j)$ or the trajectory $A(4j)$.
In the first case, since $j\ge n+l+1$,
this would imply that during the time interval $I$, agent $a$ followed an entire trajectory $B(k,v)$,
where $v$ is the starting node of $a$, for
$k\ge 2(n+l+1)$, while $b$ followed only all or a part of the trajectory $\Omega_b(2(n+l))$  consisting of repetitions of $X(2(n + l))$.
{In view of Lemma~\ref{tunel},} this would force a meeting because, by Definition~\ref{def:B},
the trajectory $B(k,v)$, for $k\ge 2(n+l+1)$ contains at least one trajectory $X(2(n+l),u)$ for every node $u$ of the graph. In the second case, in view of $j\ge n+l+1$, a meeting would be forced in a similar way, because the trajectory $A(4j,v)$, also contains at least one trajectory $X(2(n+l),u)$ for every node $u$ of the graph.

This shows that  agent $a$ has started the last atom {$M$} of its $j$th piece $\mathcal{T}_a(j)$ during the time interval $I$. Using a similar argument we prove that agent $a$ could not 
complete the first atom {$N$} of  its $(j+1)$th piece during the time interval $I$. This completes the proof. 
\end{proof}

\begin{lemma}
\label{claim2}
Let $b$ be the first agent to complete its $(2(n+l))$th fence. If the agents have not met before, then by the time agent $b$ completes its $(2(n+l))$th fence,
agent $a$ must have completed the last {atom $M$} of its $j$th piece, where $j$ is the index of agent $a$.
\end{lemma}

\begin{proof}
Suppose not. Then, in view of Lemma \ref{claim1}, during the time interval when agent $b$ follows its $(2(n+l))$th fence, the trajectory of $a$ is included in {$M$}. However, according to {Definitions~\ref{def:K} and~\ref{def:ome}}, the number of integral trajectories $X(2(n+l))$ in $\Omega_b(2(n+l))$ is {at least} $2(|A(8(2(n+l)))|+|B(4(2(n+l)))|)$. Moreover, according to Algorithm RV-asynch-poly, the number of {edge traversals} in {$M$} is less than $|B(2j)| + |A(4j)|$. So, since $j\leq 2(n+l)$ in view of Lemma~\ref{claim1}, the number of integral trajectories in $\Omega_b(2(n+l))$ is larger than the number of {edge traversals} in {$M$}. This would force a meeting. 
\end{proof}

{
\begin{lemma}
\label{claim3}
Let $b$ be the first agent to complete its $(2(n+l))$th fence. If the agents have not met before, then by the time agent $b$ completes the first atom of its segment $S_{{1}}(2(n+l)+1)$, agent $a$ must have completed its $j$th fence $\Omega_a(j)$, where $j$ is the index of agent $a$.
\end{lemma}}
{
\begin{proof}
Suppose not. Then, in view of Lemma \ref{claim2}, during the time interval when agent $b$ follows the first atom of its segment $S_{{1}}(2(n+l)+1)$, the trajectory of $a$ is included in $\Omega_a(j)$. Since the trajectory $\Omega_a(j)$ consists of repetitions of the trajectory $X(j)$ starting at the same node $v$, and while following the first atom of $S_{{1}}(2(n+l)+1)$ agent $b$ followed at least one trajectory $X(j,u)$ for any node $u$ of the graph (because $j< 2(n+l)+1$ by Lemma \ref{claim1}), this would force a meeting in view of Lemma \ref{tunel}.
\end{proof}}

\begin{lemma}
\label{lem:four}
Let $b$ be the first agent to complete its $(2(n+l))$th fence. {Let $t$ be the first time at which an agent finishes its $(2(n+l)+1)$th piece. If the agents do not meet by time $t$, then
the following properties hold, for $j$ denoting the index of agent $a$}.
\begin{itemize}
\item {\bf Property 1.} Let $t'$ be the time when agent $a$ completes 
a segment $S_{i}(j+1)$, if this segment exists. Let $t''$ be the time when agent $b$ completes 
the border $K_{{i,i+1}}(2(n+l)+1)$, if this border exists. Then $t'<t''$.
\item {\bf Property 2.} Let $t'$ be the time when agent $b$ completes a segment $S_{i}(2(n+l)+1)$, if this segment exists.
Let $t''$ be the time when agent $a$ completes 
the border $K_{{i,i+1}}(j+1)$, if this border exists. Then $t'<t''$.
\item {\bf Property 3.}
 Let $t'$ be the time when agent $a$ completes 
a border $K_{{i,i+1}}(j+1)$, if this border exists. Let $t''$ be the time when agent $b$ completes 
the first atom of the segment $S_{{i+1}}(2(n+l)+1)$, if this segment exists. Then $t'<t''$.
 \item {\bf Property 4.}
Let $t'$ be the time when agent $b$ completes a border $K_{{i,i+1}}(2(n+l)+1)$, if this border exists.
Let $t''$ be the time when agent $a$ completes 
the first atom of the segment $S_{{i+1}}(j+1)$, if this segment exists. Then $t'<t''$. 
\end{itemize}
\end{lemma}

\begin{proof}
{Assume that the agents do not meet by time $t$.}
Suppose, for contradiction, that at least one of the above 4 properties is not satisfied and let $\mu$ be the smallest value of the index $i$ for which one of these properties is not satisfied. We consider 4 cases.

{\bf Case~1.} Property~1 is false for $i=\mu$. This implies that $b$ completed its border $K_{{\mu,\mu+1}}(2(n+l)+1)$ before $a$ completed $S_{\mu}(j+1)$. Hence agent $b$ has completed $K_{{\mu,\mu+1}}(2(n+l)+1)$ while agent $a$ was following 
$S_{\mu}(j+1)$. Indeed, if agent $b$ completed $K_{{\mu,\mu+1}}(2(n+l)+1)$ before agent $a$ started $S_{\mu}(j+1)$, this would imply:
\begin{itemize}

\item if $\mu>1$ then agent $b$ started $K_{{\mu,\mu+1}}(2(n+l)+1)$ before agent $a$ has completed $K_{{\mu-1,\mu}}(j+1)$. Hence $b$ had completed $S_{\mu}(2(n+l)+1)$ before $a$ completed $K_{{\mu-1,\mu}}(j+1)$. This would imply that Property~3 is not satisfied for $\mu-1$, which contradicts the definition of $\mu$. 

\item if $\mu=1$ then agent $b$ started $K_{{1,2}}(2(n+l)+1)$ before $a$ has completed its $j$th fence $\Omega(j)$. {This is a contradiction with Lemma \ref{claim3}}


\end{itemize}

Hence agent $b$ has completed $K_{{\mu,\mu+1}}(2(n+l)+1)$ while agent $a$ was following 
$S_{\mu}(j+1)$. Similarly as before, agent $b$ has also started following $K_{{\mu,\mu+1}}(2(n+l)+1)$ while agent $a$ was following $S_{\mu}(j+1)$. 

Hence agent $b$ has followed the entire trajectory $K_{{\mu,\mu+1}}(2(n+l)+1)$ while $a$ was following $S_{\mu}(j+1)$. 
However, by Definition~\ref{def:K}, the trajectory $K_{{\mu,\mu+1}}(2(n+l)+1)$ contains $2(|B(4(2(n+l)+1))|+|A(8(2(n+l)+1))|)$ integral trajectories $X(2(n+l)+1)$. Moreover, according to Algorithm RV-asynch-poly, the number of {edge traversals} in trajectory $S_{\mu}(j+1)$ is equal to $2(|B(2(j+1))|+|A(4(j+1)|)$ which is at most $2(|B(2(2(n+l)+1))|+|A(4(2(n+l)+1))|)$ (recall that $j\leq 2(n+l)$ by Lemma \ref{claim1}). Thus, this would force a meeting because the number of integral trajectories $X(2(n+l)+1)$ in $K_{{\mu,\mu+1}}(2(n+l)+1)$ is larger than the number of {edge traversals} in $S_{\mu}(j+1)$, which is a contradiction.

{\bf Case~2.} Property 2 is false for $i=\mu$. This implies that agent $a$ completed $K_{{\mu,\mu+1}}(j+1)$ before agent $b$ completed $S_{\mu}(2(n+l)+1)$. Hence agent $a$ completed $K_{{\mu,\mu+1}}(j+1)$ while agent $b$ was following 
$S_{\mu}(2(n+l)+1)$. Indeed, if agent $a$ completed $K_{{\mu,\mu+1}}(j+1)$ before $b$ started $S_{\mu}(2(n+l)+1)$, this would imply:
\begin{itemize}
\item if $\mu>1$ then agent $a$ started $K_{{\mu,\mu+1}}(j+1)$ before agent $b$ completed $K_{{\mu-1,\mu}}(2(n+l)+1)$. Hence $a$ had completed $S_{\mu}(j+1)$ before $b$ completed $K_{{\mu-1,\mu}}(2(n+l)+1)$. This would imply that Property~4 is not satisfied for $\mu-1$, which contradicts the definition of $\mu$. 

\item if $\mu=1$ then agent $a$ started $K_{{1,2}}(j+1)$ before $b$ completed its $(2(n+l))$th fence  $\Omega(2(n+l))$ which contradicts Lemma \ref{claim2}.
\end{itemize}
Hence agent $a$ completed $K_{{\mu,\mu+1}}(j+1)$ while $b$ was following $S_{\mu}(2(n+l)+1)$. Similarly as before, agent $a$ started $K_{{\mu,\mu+1}}(j+1)$ while $b$ was
following $S_{\mu}(2(n+l)+1)$.

Hence agent $a$ has followed the entire trajectory $K_{{\mu,\mu+1}}(j+1)$ while $b$ was following $S_{\mu}(2(n+l)+1)$. 
However, by Definition~\ref{def:K}, the number of integral trajectories $X(j+1)$ (note that $X(j+1)$ is integral because $j\ge n+l+1$ in view of Lemma \ref{claim1}) in $K_{{\mu,\mu+1}}(j+1)$ is $2(|A(8(j+1))|+|B(4(j+1))|)$ which is at least $2(|A(8(n+l+2))|+|B(4(n+l+2)|)$ because $j\ge n+l+1$ in view of Lemma \ref{claim1}. Moreover, according to Algorithm RV-asynch-poly, the number of {edge traversals} in $S_{\mu}(2(n+l)+1)$ is less than $2(|A(4(2(n+l)+1))|+|B(2(2(n+l)+1))|)=2(|A(8(n+l)+4))|+|B(4(n+l)+2))|)$. Thus, this would force a meeting because the number of integral trajectories $X(j+1)$ in $K_{{\mu,\mu+1}}(j+1)$ is larger than the number of {edge traversals} in $S_{\mu}(2(n+l)+1)$, which is a contradiction.

{\bf Case~3.} Property~3 is false for $i=\mu$. This implies that agent $b$ completed the first atom of  $S_{{\mu+1}}(2(n+l)+1)$ before agent $a$ completed $K_{{\mu,\mu+1}}(j+1)$. This implies that agent $b$ completed the first atom of  $S_{{\mu+1}}(2(n+l)+1)$ while $a$ was following $K_{{\mu,\mu+1}}(j+1)$. Indeed, otherwise agent $b$ would have completed $K_{{\mu,\mu+1}}(2(n+l)+1)$ before $a$ completed $S_{{\mu}}(j+1)$ which would imply that Property 1 is false for $\mu$. This
is impossible by Case 1.

Hence agent $b$ completed the first atom of $S_{{\mu+1}}(2(n+l)+1)$ while $a$ was following $K_{{\mu,\mu+1}}(j+1)$.  For the same reasons agent $b$ also started the first atom of $S_{{\mu+1}}(2(n+l)+1)$ while $a$ was following $K_{{\mu,\mu+1}}(j+1)$.  Then while $a$ was following $K_{{\mu,\mu+1}}(j+1)$,  agent $b$ either followed entirely the trajectory $A(8(n+l)+4)$ or followed entirely the trajectory $B(4(n+l)+2)$.  Consequently, in view of Definitions ~\ref{def:A} and~\ref{def:B}, agent $b$ must have followed trajectory $X(j+1,u)$ for every node $u$ of the graph at least once  (because $j\le 2(n+l)$ by Lemma \ref{claim1}). Since $K_{{\mu,\mu+1}}(j+1)$ consists of repeating $X(j+1,v)$ for the same node $v$,  agents would meet in view of Lemma \ref{tunel}, which is a contradiction.

{\bf Case~4.} Property~4 is false for $i=\mu$. This implies that agent $a$ completed the first atom of  $S_{{\mu+1}}(j+1)$ before agent $b$ completed 
$K_{{\mu,\mu+1}}(2(n+l)+1)$. This implies that agent $a$ completed the first atom of  $S_{{\mu+1}}(j+1)$ while $b$ was following $K_{{\mu,\mu+1}}(2(n+l)+1)$. Indeed, otherwise agent $a$ would have completed $K_{{\mu,\mu+1}}(j+1)$before agent $b$ completed $S_{{\mu}}(2(n+l)+1)$ which would imply that Property 2 is false for $\mu$. This
is impossible by Case 2.

Hence agent $a$ completed the first atom of $S_{{\mu+1}}(j+1)$ while $b$ was following  $K_{{\mu,\mu+1}}(2(n+l)+1)$. For the same reasons agent $a$ also started the first atom of  $S_{{\mu+1}}(j+1)$ while $b$ was following $K_{{\mu,\mu+1}}(2(n+l)+1)$.  Then while $b$ was following
$K_{{\mu,\mu+1}}(2(n+l)+1)$, agent $a$ either followed entirely the trajectory $A(4(j+1))$ or followed entirely the trajectory $B(2(j+1))$.  Consequently,  in view of Definitions ~\ref{def:A} and~\ref{def:B},  agent $a$ must have followed trajectory $X(2(n+l)+1,u)$ for every node $u$ of the graph at least once (because  $j\ge n+l+1$ by Lemma \ref{claim1}). Since $K_{{\mu,\mu+1}}(2(n+l)+1)$ consists of repeating $X(2(n+l)+1,v)$ for the same node $v$,  agents would meet in view of Lemma \ref{tunel}, which is a contradiction.
\end{proof}

\begin{theorem}\label{main}
There exists a polynomial $\Pi(x,y)$, non decreasing in each variable, such that
if two agents with labels $L_1$ and $L_2$ execute Algorithm RV-asynch-poly in a graph of size $n$, then their meeting is guaranteed by the time one of them performs $\Pi(n,\min(|L_1|,|L_2|))$ edge traversals. 
\end{theorem}

\begin{proof}
Let $m=\min(|L_1|,|L_2|)$. Let $a$ be the agent with label $L_1$ and let $b$ be the agent with label $L_2$.
Let $M_a$ be the modified label of agent $a$ and let $M_b$ be the modified label of agent $b$. Let $l$ be the length of the shorter of labels $M_a$,  $M_b$.
Hence $l=2m+2$.
As observed before, the modified label of one agent cannot be a prefix of the modified label of the other. Hence there exists an integer {$l \geq \lambda > 1$}, such that 
the $\lambda$th bit of $M_a$ is different from the $\lambda$th bit of $M_b$.
{Let $t$ be the first time at which an agent finishes its $(2(n+l)+1)$th piece.
By Lemma \ref{lem:four},  if the agents have not met by time $t$,} then one of them cannot have completed the first atom of  $S_\lambda(k_1)$ as long as the other agent has not completed $K_{\lambda-1,\lambda}(k_2)$ {(i.e. started $S_\lambda(k_2)$)}, {for some $2(n+l)+1\ge k_1,k_2 \ge n+l+2$. (Since $k_1,k_2 \ge n+l+2$ and $l \geq \lambda > 1$, these objects must exist.)}

First suppose that the $\lambda$th bit of $M_a$ is 1. There are two possible cases.

\begin{itemize}
\item agent $a$ follows the entire trajectory $B(2(j+1))$ while agent $b$ is following $S_{\lambda}(2(n+l)+1)=A(8(n+l)+4)^2$.

Since $j\ge n+l+1$, by Definition~\ref{def:B} the trajectory $B(2(j+1))$ contains $2|A(8j+8)|\ge 2|A(8(n+l+1)+8)|$ 
integral trajectories $Y(2(j+1))$. Moreover, according to Algorithm RV-asynch-poly, the number of {edge traversals} in $S_{\lambda}(2(n+l)+1)$ is $2|A(8(n+l)+4)|$. So, the trajectory $B(2(j+1))$ contains more integral trajectories $Y(2(j+1))$ than there are {edge traversals} in $S_{\lambda}(2(n+l)+1)$, hence there is a meeting.

\item agent $b$ follows the entire trajectory $A(4(2(n+l)+1))$ while agent $a$ is following  $S_{\lambda}(j+1)= B(2(j+1))^2$. 

The trajectory $S_{\lambda}(j+1)$ consists of repetitions of $Y(2(j+1),v)$ for some node $v$. 
Since by Lemma \ref{claim1}, $j\le 2(n+l)$, the trajectory $A(4(2(n+l)+1))$, contains $Y(2(j+1),u)$ for every node $u$ of the graph, which implies a meeting by Lemma \ref{tunel}.
\end{itemize}

Next suppose that the $\lambda$th bit of $M_a$ is 0. There are two possible cases.

\begin{itemize}

\item agent $a$ follows the entire trajectory $A(4(j+1))$ while agent $b$ is following $S_{\lambda}(2(n+l)+1)=B(2(2(n+l)+1))^2$. 

The trajectory $S_{\lambda}(2(n+l)+1)$ consists of repetitions of {$Y(4(n+l)+2,v)$} for some node $v$.  Since by Lemma \ref{claim1}, $j\ge n+l+1$, the trajectory
$A(4(j+1))$, contains $Y(4(n+l)+2,u)$  for every node $u$ of the graph, which implies a meeting by Lemma \ref{tunel}.

\item agent $b$ follows the entire trajectory $B(2(2(n+l)+1))$ while agent $a$ is following $S_{\lambda}(j+1)=A(4(j+1))^2$. 

By Definition~\ref{def:B} the trajectory $B(2(2(n+l)+1))$ contains $2|A(16(n+l)+8)|$ integral trajectories $Y(2(2(n+l)+1))$. Moreover, since $j\le 2(n+l)$, the number of {edge traversals} in $S_{\lambda}(j+1)$ is $2|A(4(j+1))|$ i.e., at most $2|A(8(n+l)+4)|$. So, the number of integral trajectories $Y(2(2(n+l)+1))$ in $B(2(2(n+l)+1))$ is larger than the number of {edge traversals} in $S_{\lambda}(j+1)$, hence there is a meeting.
\end{itemize}

{Hence in all cases agents meet by the time when the first of the agents completes its $(2(n+l)+1)$th piece}. Now the proof can be completed by the following estimates which are a consequence of the formulation of the algorithm and of the definitions of respective trajectories.
 (Recall that $P$ is the
polynomial describing the number of edge traversals in the trajectory obtained by Reingold's procedure.)

For any $v$, $|X(k,v)| \leq X^*_k=2P(k)+1$.

For any $v$, $|Q(k,v)| \leq Q^*_k= \sum_{i=1}^k X^*_i$.

For any $v$, $|Y(k,v)| \leq Y^*_k=2P(k) \cdot Q^*_k$.

For any $v$, $|Z(k,v)| \leq Z^*_k= \sum_{i=1}^k Y^*_i$.

For any $v$, $|A(k,v)| \leq A^*_k=2P(k) \cdot Z^*_k$.

{For any $v$, $|B(k,v)| \leq B^*_k=2A^*_{4k} \cdot Y^*_k$.}

{For any $v$, $|K(k,v)| \leq K^*_k=2(B^*_{4k}+A^*_{8k}) \cdot X^*_k$.}

For any $v$, $|\Omega(k,v)| \leq \Omega^*_k=(2k-1)K^*_{k} \cdot X^*_k$.

For every integer $k>0$, let $T^*_k$ denote the number of nodes in a piece in iteration $k$ of the repeat loop in Algorithm RV-asynch-poly. Let $N=2(n+l)+1$. Recall that $l=2m+2$. 
 We have $T_k^* \le  N(2A_{4k}^* + 2B_{2k}^* + K_k^*)$.
 {For any agent, the length of the trajectory it follows by the time it
 completes the $(2(n+l)+1)$th piece is at most $\sum_{k=1}^N(T^*_k+\Omega^*_k)$.} Let $\Pi(n,m)=\sum_{k=1}^N(T^*_k+\Omega^*_k)$.
 It follows from the above discussion that agents must meet by the time one of them performs $\Pi(n,m)$ edge traversals.
 Since $T^*_k$ and $\Omega^*_k$ are polynomials in $k$, while $N$ and $l$ are polynomials in $n$ and $m$, the function $\Pi(n,m)$ is a polynomial.
Since the polynomial $P(k)$ is non-decreasing, $\Pi$ is non-decreasing in each variable. This completes the proof.
\end{proof}

\section{Applications: solving problems for multiple asynchronous agents}

In this section we apply our polynomial-cost rendezvous algorithm for asynchronous agents to solve four basic distributed problems
involving multiple asynchronous agents in unknown networks. Agents solve these problems by exchanging information during their meetings.
The scenario for all the problems is the following. There is a team of $k>1$ agents having distinct integer labels, located at different nodes of an unknown network. 
The adversary wakes up some of the agents at possibly different times. A dormant agent is also woken up by an agent that visits its starting node, if such an agent exists.
As before, each agent knows a priori only its own label. Agents do not know the size of the team and, as before, have no
a priori knowledge concerning the network.
The assumptions concerning the movements of agents remain unchanged. We only need to add a provision in the model specifying what happens when agents meet. (For rendezvous, this was the end of the process.) This addition is very simple. 
When (two or more)  agents meet, they notice this fact and can exchange all previously acquired information. 
However, if the meeting is inside an edge,
they continue the walk prescribed by the adversary until reaching the other end of the current edge. New knowledge acquired at the meeting
can then influence the choice of the subsequent part of the routes constructed by each of the agents. It should be noted that the possibility
of exchanging all current information at a meeting is formulated only for simplicity. In fact, during a meeting, our algorithm prescribes the exchange of only at most
$k$ labels of other agents that the meeting agents have already heard of, their initial values in the case of the gossiping problem, and a constant number of control bits.

We now specify the four problems that we want to solve:
\begin{itemize}
\item
{\tt team size}: every agent has to output the total
number $k$ of agents; 
\item{\tt leader election}: all agents have to output the label of a single agent, called the leader;
\item
{\tt perfect renaming}: all agents
have to adopt new different labels from the set $\{1,\dots,k\}$, where $k$ is the number of agents;  
\item
{\tt gossiping}:  each agent has initially a piece of information (value) and all agents have to output all the values; thus agents have to exchange
all their initial information.
\end{itemize} 
The cost of a solution of each of the above problems is the total number of edge traversals by all agents until they output the solution.
Using our rendezvous algorithm we solve all these problems at cost polynomial in the size of the graph and in the smallest length of all labels of participating agents. 

Let us first note that accomplishing all the above tasks is a consequence of solving the following problem: 
at some point each agent acquires the labels of all the agents {\em and is aware} of this fact. We call this more general problem Strong Global Learning (SGL), where the word ``strong''  emphasizes awareness of the agents that learning is accomplished.\footnote{Notice that the assumption that the number $k$ of agents is larger than 1 is necessary. For a single agent neither SGL nor any of the above mentioned problems can be solved. Indeed,
for example in an oriented ring of unknown size (ports 0,1 at all nodes in the clockwise direction), a single agent cannot realize that it is alone.}
Indeed, if each agent gets the labels of all the agents {\em and is aware} of it, then each agent can count all agents, thus solving
{\tt team size}, each agent can output the smallest label as that of the leader, thus solving {\tt leader election},  each agent can adopt the new label
$i$ if its original label was $i$th in increasing order among all labels, 
thus solving {\tt perfect renaming}, and each agent can output all initial values,
thus solving {\tt gossiping},  if we append in the algorithm for SGL the initial value to the label of each agent.

Hence it is enough to give an algorithm for the SGL problem, working at cost polynomial in the size of the graph 
and in the smallest length of all labels of participating agents. This is the aim of the present section. Notice that this automatically solves all
distributed problems that depend only on acquiring by all agents the knowledge of all labels and {\em being aware of this fact}.
(The above four problems are in this class.) We stress this latter requirement, because it is of crucial importance. Note, for example, that none of the above four problems can be solved even if agents
eventually learn all
labels but are never aware of the fact that no other agents are in the network. 
This detection requirement is non-trivial to achieve: recall that agents have 
no a priori bound on the size of the graph or on the size of the team.

We now describe Algorithm SGL solving the SGL problem at cost polynomial in the size of the graph 
and in the smallest length of all labels of participating agents.  In this description we will use procedure RV-ASYNCH-POLY$(L)$ to denote Algorithm RV-asynch-poly
as executed by an agent with label $L$. 

\vspace*{0.3cm}

\noindent
{\bf Algorithm SGL}

{We will define three states in which an agent can be. These states are {\em traveller}, {\em explorer}
and {\em ghost}. Transitions between states depend on the history of the agent, and more specifically on comparing the labels exchanged during meetings.} 

{The high-level idea of the algorithm is the following. An agent $a$ with label $L$ wakes up in state {\em traveller} and executes procedure RV-ASYNCH-POLY$(L)$ until the first meeting when it meets either agents that are not in state {\em explorer}, or agents having heard of some label smaller than $L$ (below we explain what ''having heard of'' exactly means via the notion of {\em bag}).}
{Then, depending on the comparison of labels of the agents it meets or the labels that have been heard of by the agents it meets, it transits either to state {\em ghost} or to state {\em explorer}. In the first case it terminates its current move and stays idle.
In the second case it simulates procedure $ESST$ by using an agent in state {\em ghost} as token and learns a polynomial upper bound $E(n)$ on the size $n$ of the graph. Then it resumes procedure RV-ASYNCH-POLY$(L)$, from where it interrupted it when leaving state {\em traveller}, until it performes the $\Pi(E(n),L)$ edge traversals of RV-ASYNCH-POLY$(L)$ or it hears of another agent having a smaller label than $L$.}

{If the agent is informed about the existence of a label smaller than $L$ before executing the $\Pi(E(n),L)$ edge traversals of RV-ASYNCH-POLY$(L)$, it switches to state {\em ghost}: in fact we will prove that this kind of situation occurs for all explorers having a label different from $M$ (where $M$ is the smallest label among the participating agents) and after at most a number of edge traversals polynomial in $n$ and $|M|$.} 

{Otherwise, the agent ends up executing the $\Pi(E(n),L)$ edge traversals of RV-ASYNCH-POLY$(L)$ (we will show that this occurs only when $L=M$).
At this point, there are no longer agents in state {\em traveller} and an execution of Reingold's procedure
followed by a complete backtrack of the trajectory resulting from this procedure permits the agent to learn all labels of participating agents as well as to convey this knowledge to all agents in state {\em ghost}.
All other agents will in turn get this knowledge from these agents.}

Below we specify what an agent $a$ with label $L$
does in each state and how it transits from state to state. Each agent has a set variable $W$, called its {\em bag}, initialized to $\{L\}$, where $L$ is its label. At each point of the execution of the algorithm  the value of the bag is the set of labels of agents that $a$ has been informed about {(the bag of an agent is the set of labels it has heard of)}. More precisely, during any meeting of $a$ with agents whose current values of their bags are $W_1,W_2,\dots, W_i$,
respectively, agent $a$ sets the value of its bag $W$ to $W\cup W_1 \cup W_2 \cup \dots \cup W_i$.  Notice that since each bag can be only incremented, the number of updates of each bag
is at most $k-1$, where $k$ is the number of agents.

\vspace*{0.3cm}

\noindent
State {\em traveller}.

{The agent $a$ wakes up in this state and starts executing procedure RV-ASYNCH-POLY$(L)$  until the first meeting. Suppose the first meeting is with a set $Z$ of agents ($a\notin Z$). If there is an agent in $Z$ having a bag which includes a value smaller than $L$ then agent $a$ transits to state {\em ghost}.} 

{Otherwise, if $Z$ contains an agent in state
{\em ghost} or {\em traveller}, then agent $a$ transits to state {\em explorer} and the smallest agent in set $Z$ which is not an explorer, say agent $b$, will play the role of the token of agent $a$ in order to simulate procedure $ESST$ (refer to state {\em explorer}). Note that if agent $b$ is not in state {\em ghost} when it meets agent $a$ then its transits to this state while $a$ transits to state {\em explorer}}.

{In all the other cases, agent $a$ remains in state {\em traveller} and continues executing procedure RV-ASYNCH-POLY$(L)$ until the next meeting.}

\vspace*{0.3cm}

\noindent
State {\em ghost}.

{Agent $a$ completes the traversal of the current edge and remains idle at its extremity forever.
As soon as it gets the information (from some agent in state {\em explorer}) that its current bag contains all labels of participating agents, agent $a$ outputs the value of its bag.}

\vspace*{0.3cm}

\noindent
State {\em explorer}.

{When agent $a$ transits to this state, it has just met an agent $b$ in state {\em ghost} (or which has just transited to state {\em ghost}), that $a$ considers as its token.
The actions of agent $a$ are divided into three phases. If agent $a$ transited to state {\em explorer} while traversing an edge, Phase~1 starts as soon as this edge traversal is done. Otherwise, Phase~1 starts immediately.}

\noindent \underline{Phase 1}.

{If agent $a$ transited to state {\em explorer} from a node $v$, agent $a$ performs procedure $ESST$ with its token which stays idle at $v$ on the extended edge $u-v$ (where $u$ is some node adjacent to $v$). Otherwise, agent $a$ transited to state {\em explorer} while traversing an edge $u-v$ from $u$ to $v$. In this latter case, agent $a$ also performs procedure $ESST$ with its token located on the extended edge $u-v$ (according to state {\em ghost}, the token remains on this extended edge forever)}.

{After completing Phase~1, agent $a$ has visited all the nodes of the graph and it knows a polynomial upper bound on the size $n$ of the network: this upper bound, denoted $E(n)$, is the cost of the entire execution of $ESST$ previously made by agent $a$ (refer to Theorem~\ref{theo:est}).} 

\noindent  \underline{Phase 2}.

{Let $T$ be the trajectory made by the agents during Phase~1. Agent $a$ backtracks to node $v$ using the trajectory $\overline{T}$. Then, knowing the polynomial upperbound $E(n)$ on the size of the graph, agent $a$ resumes the execution of procedure RV-ASYNCH-POLY$(L)$ (from where it interrupted it when transiting from state {\em traveller} to state {\em explorer}) and executes it until it made $\Pi(E(n),L)$ edge traversals of RV-ASYNCH-POLY$(L)$. More precisely, procedure RV-ASYNCH-POLY$(L)$ was interrupted either at node $v$ just after $a$ completed the first $k$ edge traversals of the procedure or on edge $u-v$ while $a$ was walking from $u$ to $v$, executing the $k$-th edge traversal of the procedure. In the first case, agent $a$ resumes the execution of procedure RV-ASYNCH-POLY$(L)$ from node $v$ by making the $(k+1)$-th edge traversal, the $(k+2)$-th edge traversal, etc., until the $\Pi(E(n),L)$-th edge traversal. In the second case, the agent does the same but by resuming the procedure from $u$ by executing the $k$-th edge traversal of RV-ASYNCH-POLY$(L)$ (instead of starting from node $v$ with the $(k+1)$-th edge traversal): to do so, the agent first moves from node $v$, where it is currently located, to node $u$.}

{Whenever $Min(W)<L$, agent $a$ aborts Phase~2 as soon as it is at a node and switches to Phase~3.}

\noindent  \underline{Phase 3}.

{Let $s$ be the node where agent $a$ is located at the beginning of Phase~3.
If $Min(W)<L$ then agent $a$ seeks to meet its token (i.e. agent $b$) by applying $R(E(n),s)$. Once the meeting occurs, if agent $b$ has already output its bag then agent $a$ does the same. Otherwise, agent $a$ transits to state {\em ghost}.}

{If $Min(W)=L$, we know that agent $a$ has carried out the execution of Phase~2 until its term without aborting it prematurely. We will show that
at this point  there does not remain any agent that is either dormant or in state {\em traveller}. The labels of all remaining agents are in the union of bags of all agents currently in state {\em ghost}. In this case, agent $a$  performs $R(E(n),s)$ followed by a complete backtrack $\overline{R(E(n),s)}$.
After the first trajectory $R(E(n),s)$ the agent has in its bag the labels of all participating agents.
During the second trajectory $\overline{R(E(n),s)}$, all these labels are transmitted to all agents in state {\em ghost}, together with the information that this is the set of all labels. After completing the second trajectory
agent $a$ outputs the value of its bag.}

\begin{theorem}\label{sgl} 
Upon completion of Algorithm SGL, each agent outputs the set of labels of all participating agents. The total cost of the algorithm is polynomial 
in the size of the graph 
and in the smallest length of all labels of participating agents.
\end{theorem}




\begin{proof}
{Let $m$ be the agent having the smallest label, denoted $M$, among all the participating agents. The argument is split in proofs of two claims.}

\vspace*{0.3cm}
\noindent
{{\bf Claim~1.} By applying Algorithm SGL in a graph of size $n$, every agent makes a number of edge traversals polynomial in $n$ and $|M|$.}

{To prove this claim, consider an agent $a$ with label $L$ {(label $L$ can be any label among those that are carried by the agents circulating in the graph: In particular, if $L=M$, agent $a$ corresponds to agent $m$)}. If agent $a$ never wakes up, it makes no edge traversals. So, let us focus on the case where it eventually wakes up. Upon waking up the agent is in state {\em traveller} and starts executing  procedure RV-ASYNCH-POLY$(L)$. In view of Theorem \ref{main},  
by the time the agent performs $\Pi(n,|M|)$ edge traversals, it must meet some agent that is in state {\em traveller} or in state {\em ghost}, or some agent with a bag containing a label smaller than $L$ (if $L\ne M$).} ({Indeed, in the case $L=M$ note that during this time interval if agent $a$ (which corresponds to agent $m$ in this case) does not meet any agent in state {\em traveller}, it must meet an agent in state {\em ghost} because there is an agent in state ghost located in each edge at which an agent transited from state {\em traveller} to state {\em explorer}. In the case $L\ne M$ (i.e., when agent $a$ and agent $m$ are different), note that} during this time interval, agent $m$ is idle or is executing procedure RV-ASYNCH-POLY$(M)$ as a traveller, or an agent in state {\em ghost} playing the role of the token of $m$ is located in the edge at which $m$ stopped the execution of RV-ASYNCH-POLY$(M)$ to transit to state {\em explorer}. So if agent $a$ has not met another agent in state {\em traveller} or in state {\em ghost} or some agent with a bag containing a label smaller than $L$ before it makes the $\Pi(n,|M|)$th edge traversal of RV-ASYNCH-POLY$(L)$, it must meet agent $m$ or the token of agent $m$ while making the $\Pi(n,|M|)$th edge traversal of RV-ASYNCH-POLY$(L)$.) At this meeting, agent $a$ transits either to state
{\em ghost} or to state {\em explorer}. In the first case agent $a$ does not perform any further edge traversals and the claim follows in that case. So consider the second case, when agent $a$ transited to state {\em explorer}.
In this case  agent $a$  uses at most 
{$T(ESST(n))$ edge traversals in Phase 1 in order to perform procedure $ESST$}.
After completing Phase 1 agent $a$ knows a polynomial upper bound $E(n)$ on the size of the graph. 

{In state {\em explorer} agent $a$ starts Phase 2 by executing a complete backtrack of the trajectory made by the agent in Phase~1 which also costs at most 
{$T(ESST(n))$} edge traversals. Then after at most one extra edge traversal, agent $a$ resumes the execution of procedure RV-ASYNCH-POLY$(L)$ from where it interrupted it (when leaving state {\em traveller}) until it made the $\Pi(E(n),|L|)$th edge traversal of RV-ASYNCH-POLY$(L)$ or as soon as $Min(W)<L$, where $W$ is the bag of agent $a$. However, notice that if $a\ne m$ then agent $a$ cannot go beyond the execution of the $\Pi(n,|M|)$th edge traversal of RV-ASYNCH-POLY$(L)$ (and thus every explorer different from $m$ aborts Phase~2 having a bag with a value smaller than its own label). Indeed, if at the time when procedure RV-ASYNCH-POLY$(L)$ is resumed we  have $Min(W)=L$, then the token of $a$ has not met agent $m$ executing RV-ASYNCH-POLY$(L)$ as a traveller. Hence according to Algorithm $SGL$ and Theorem \ref{main}, if agent $a$ does not meet an agent with a bag containing a value smaller than $L$ before the execution of the $\Pi(n,|M|)$th edge traversal of RV-ASYNCH-POLY$(L)$ then it meets either agent $m$ or the token of $m$ while executing the $\Pi(n,|M|)$th edge traversal of 
RV-ASYNCH-POLY$(L)$ which immediately makes $Min(W)$ smaller than $L$.}

{Thus Phase~2 costs at most {$T(ESST(n))+1+\Pi(n,|M|)$} edge traversals if agent $a$ is different from $m$, and at most {$T(ESST(n))+1+\Pi(E(n),|M|)$} for agent $m$, which leads to an upper bound of {$T(ESST(n))+1+\Pi(E(n),|M|)$ for any agent}.}

{Since in Phase~3, an explorer executes at most $2P(E(n))$ edge traversals, the total number of edge traversals performed by any agent can be upper-bounded by  {$\Pi(n,|M|)+2T(ESST(n))+1+\Pi(E(n),|M|)+2P(E(n))$} which is polynomial in $n$ and $|M|$. Hence the claim is proven.}

\vspace*{0.3cm}

\noindent
{{\bf Claim~2.} By applying Algorithm SGL in a graph of size $n$, every agent eventually outputs its bag. Moreover, when a bag is output, it contains the labels of all the participating agents.}

{To prove the claim, first note that only agent $m$ ends up executing $R(E(n),s)\overline{R(E(n),s)}$ for some node $s$ in Phase~3 of state {\em explorer}. Indeed, a necessary condition, for agent with label $L$ to execute this, is that its bag does not contain a label smaller than $L$. However, as mentioned in the proof of Claim~1, at the end of Phase~2 every explorer different from $m$ has its bag containing a value smaller than its own label. Moreover, from Algorithm $SGL$ we know that an agent, say $e$, transits to state {\em explorer} by the time when the first woken up agent leaves state {\em traveller}, and thus in view of Theorem~\ref{theo:est}, agent $m$ is woken up by the time agent $e$ finishes executing Phase~1 of state {\em explorer}. Finally, agent $m$ never transits to state {\em ghost} and eventually executes $R(E(n),s)\overline{R(E(n),s)}$ in Phase~3 of state {\em explorer} because its bag cannot include a label smaller than $M$.}

{Suppose that at the end of the execution of Phase~2 by agent $m$ at time $t$, there remains an agent $x$ that is either still dormant or in state {\em traveller}. In particular this means that agent $m$ does not meet agent $x$ during the execution of the $\Pi(E(n),|M|)$ edge traversals of RV-ASYNCH-POLY$(L)$, first as a {\em traveller} and then as an {\em explorer}. From Theorem~\ref{main} and Algorithm $SGL$ it follows that the token of $m$ necessarily meets agent $x$ by time $t$. However, by meeting the token of $m$, which has in its bag label $M$, agent $x$ must transit to state {\em ghost} by time $t$ according to Algorithm $SGL$, which is a contradiction. Hence at time $t$ all the agents different from $m$ are in state {\em ghost} or {\em explorer} and all the labels of participating agents are in the union of the bags of agents in state {\em ghost}.}

{After round $t$, according to Phase~3 of state {\em explorer}, agent $m$ performs $R(E(n),s)\overline{R(E(n),s)}$. By the end of $R(E(n),s)$, agent $m$ must meet all agents that
were in state {\em ghost} at time $t$, as these agents never enter a different edge from the one where they transited to state {\em ghost}. Consequently, by the end of $R(E(n),s)$, the bag of agent $m$
contains the labels of all agents, and $m$ is aware of this fact. 
During the execution of $\overline{R(E(n),s)}$, agent $m$ transmits its bag to all agents currently in state {\em ghost} together with the information
that this bag contains all labels. 
This permits all agents currently in state {\em ghost} to output the value of their bag which now contains all labels. 
Upon completion of Phase~3, agent
$m$ outputs the value of its bag which contains all labels.}

{To conclude the proof of this claim, we have to argue that each agent (different from $m$) that is in state {\em explorer} at time $t$, also eventually outputs its bag with the labels of all participating agents. This is the case because, as mentioned before, these agents end up transiting to state {\em ghost} by executing Phase~3 of state {\em explorer}. Indeed, if such an agent transits to state {\em ghost} by the end of the execution of Phase~3 by agent $m$, it will get the information that its bag contains all the labels and can be output, either from agent $m$ or from its token, when transiting to state {\em ghost}. Otherwise, it transits to state {\em ghost} after the end of the execution of Phase~3 by agent $m$, and thus it gets this final information from its token when transiting from state {\em explorer} to state {\em ghost}, which proves the claim.}

{The theorem follows from Claims~1 and~2.}
\end{proof}

\section{Conclusion}

We presented an algorithm for asynchronous rendezvous of agents in arbitrary finite connected graphs, working at cost polynomial in the size  of the graph and in the length
of the smaller label. In \cite{CLP}, where the exponential-cost solution was first proposed, the authors stated the following question:
\begin{quotation}
Does there exist a deterministic asynchronous rendezvous algorithm,
working for all connected finite unknown graphs,
with complexity polynomial in the labels of the agents and in the
size of the graph?
\end{quotation}
Our result gives a strong positive answer to this problem: our algorithm is polynomial in the {\em logarithm} of the smaller label and in the size of the graph.

In this paper we did not make any attempt at optimizing the cost of our rendezvous algorithm, the only concern was to keep it polynomial. Cost optimization seems to be a very challenging problem. Even finding the optimal cost of exploration of unknown graphs of known size is still open, and this is a much simpler problem, as it is equivalent to rendezvous of two agents one of which is inert. 

We also applied our rendezvous algorithm to solve four fundamental distributed problems in the context of multiple asynchronous mobile agents.
The cost of all solutions is polynomial in the size of the graph and in the length of the smallest of all labels.


\bibliographystyle{plain}


\end{document}